\documentclass[11pt]{amsart}
\usepackage{amsmath,amssymb,amsfonts,amsbsy}
\usepackage{tikz-cd}
\usetikzlibrary{shapes}
\usepackage{amsbsy}
\usepackage{amscd}
\usepackage{wasysym}
\usetikzlibrary{matrix,arrows,decorations.pathmorphing}
\usepackage{graphicx}
\usepackage{float}
\usepackage{setspace}
\usepackage[margin=1.0in]{geometry}

\newtheorem{thm}{Theorem}[section]
\newtheorem{lem}[thm]{Lemma}
\newtheorem{cor}[thm]{Corollary}
\newtheorem{prop}[thm]{Proposition}

\theoremstyle{definition}
\newtheorem{example}[thm]{Example}
\theoremstyle{definition}
\newtheorem{defn}[thm]{Definition}
\theoremstyle{definition}
\newtheorem{remark}[thm]{Remark}

\newcommand{\rmv}[1]{}

\newcommand{\C}{\mathbb{C}}
\newcommand{\N}{\mathbb{N}}

\newcommand{\R}{\mathbb{R}}

\begin{document}

\title[Factorizations of Quantum Channels]{Factorizations of Quantum Channels}
\author[J. Levick,]{Jeremy Levick$^{1}$}

\address{$^1$Department of Mathematics \& Statistics, University of Cape Town, Cape Town, South Africa}

\begin{abstract}
In this paper, we discuss the relationship between the geometry of correlation matrices and certain quantum channels defined by means of such matrices: the Schur product channels. In particular, we find some relationships between complementary channels of Schur product channels, the geometry of the set of correlation matrices, and the notion of factorizability of a Schur product channel. We mention how these results generalize to non-Schur product channels, and we discuss some examples, as well as mention the relationship to notions such as quantum privacy.
\end{abstract}

\subjclass[2010]{47L90, 81P45, 94A40}

\keywords{Schur product, correlation matrix, quantum channel, factorization.}

\begin{abstract}
We find a connection between the existence of a factorization of a quantum channel and the existence of low-rank solutions to certain linear matrix equations. Using this, we show that if a quantum channel is factorized by a direct integral of factors, it must lie in the convex hull of quantum channels which are factorized respectively by the factors in the direct integral. We use this to characterize some non-trivial extreme points in the set of factorizable quantum channels and give an example.
\end{abstract}

\maketitle

\section{Introduction}
Quantum channels are well-studied in quantum information theory as the quantum analogue of classical channels. Mathematically, they are represented by completely positive, trace-preserving maps on von Neumann algebras, usually matrix algebras. It is well-known that, by the Stinespring dilation theorem \cite{stinespring}, a quantum channel $\Phi$ can be dilated so that it is the restriction of the action of a unitary operation on a larger space:
$$\Phi(X) = (\mathrm{id}\otimes \mathrm{Tr})\bigl(U(X\otimes vv^*)U^*\bigr)$$ for some $v \in \C^p$ and $U\in M_{np}(\C)$ if $\Phi:M_n(\C)\rightarrow M_n(\C)$. Physically speaking, with $\C^n$ as our system of interest, we consider $\C^p$ as an environment. Initially, the system of interest and the environment are uncoupled, hence we can write the initial state as $X\otimes \rho$, and the environment is in a pure state, $\rho = vv^*$. Then, the action of $\Phi$ on the system of interest corresponds to evolving the joint system-and-environment unitarily, and restricting attention to the system of interest.\\
An obvious question to ask is whether the state $\rho$ must be a pure state in this representation; is there some classification of the channels for which we can choose $\rho = \sum_i p_i v_iv_i^*$. In particular, we can ask whether a channel has a representation of the form 
\begin{equation}\label{fact}\Phi(X) = (\mathrm{id}\otimes \mathrm{Tr})\bigl(U(X\otimes I)U^*\bigr)\end{equation} for $I\in M_p(\C)$ for some $p$. It turns out that not all channels admit a representation of the form \ref{fact}. Those that do are called \emph{factorizable}; the name comes from the fact that if such a form exists for $\Phi$, then $\Phi$ can be factored as $$\Phi = \alpha^*\circ \beta$$ for $\alpha,\beta : M_n(\C)\rightarrow M_p(\C)$ both isometries \cite{haagerup}. As a trivial example, notice that if $\Phi(X) = UXU^*$ for a unitary $U$, then $\Phi$ is factorizable by means of $\C$:
$$\Phi(X) = (\mathrm{id}\otimes \mathrm{Tr})\bigl( U(X\otimes 1)U^*\bigr) = 1UXU^*.$$

Factorizable channels admit dilations that are also quantum channels \cite{haagerup}, and non-factorizable channels provide counterexamples to the Asymptotic Quantum Birkhoff conjecture \cite{haagerup}. Any extreme point of the set of quantum channels on $M_n(\C)$ that is not a unitary adjunction is non-factorizable, so they are also interesting for the purpose of understanding the geometry of the set of quantum channels on $M_n(\C)$. \\
In fact, the set of factorizable channels is convex, and one of the goals of this paper is to better understand the convex structure of the set of factorizable channels on $M_n(\C)$. To this end, we exhibit an example of a quantum channel that is extreme in the set of factorizable quantum channels, but is not a unitary adjunction. We also study how the factorizations of a channel $\Phi$ relate to factorizations for the extreme factorizable channels that contain $\Phi$ in their convex hull.\\
A special class of quantum channels are those that act diagonally; this corresponds to taking the Schur product of each input with some fixed correlation matrix--a positive semidefinite matrix with $1$s down the diagonal. Because of their simpler structure, these Schur product channels are easy to study as test-cases. The geometry of the set of $n\times n$ correlation matrices is intimately bound up with factorizability of the associated Schur product channels. Lastly, we mention that the study of which Schur product channels are factorizable connects with work on Connes' embedding problem \cite{haagerup}\cite{dykema}, and is interesting in its own right. \\
The paper is organized as follows: in Section $1$, we introduce quantum channels and Schur product maps, and give the necessary background. Section $2$ introduces factorizability and various equivalent conditions for a channel $\Phi$ to admit a factorization. In Section $3$ we establish a connection with linear matrix inequalities, and show that whether or not a channel is factorizable is equivalent to the question of deciding whether a linear matrix inequality over a certain operator system has solutions of sufficiently low rank. In the final section, we use the results of the previous sections to show that if a channel $\Phi$ is factorized by means of a direct integral of factor von Neumann algebras, then $\Phi$ lies in the convex hull of channels which are respectively factored by the individual factors; we use this to obtain a sufficient condition for a channel to be extreme in the set of factorizable channels, and provide an example. \\
Our notation is, we hope, fairly transparent, but we outline our usage: $M_{n,m}(\C)$ is the $n\times m$ complex matrices, and $M_n(\C):=M_{n,n}(\C)$. We use $^*$ to indicate the conjugate transpose of a matrix, $^{\dagger}$ to indicate the Hermitian adjoint of a linear map on $M_n(\C)$; so 
$\Phi^{\dagger}(X)$ is distinct from $\Phi(X)^*$. If $X\in M_{n,m}(\C)$ then $X^T$ is its transpose, and $\overline{X}$ is the entry-wise complex conjugate of $X$.\\
$X \succeq 0$ means $X$ is positive semidefinite, and $X\succeq Y$ means $X-Y\succeq 0$; $X\succ 0$ means $X$ is positive definite. The set of $n\times n$ Hermitian matrices is $S_n$, and the set of $n\times n$ positive semidefinite matrices is $S_n^+$; the set of $n\times n$ positive definite matrices is $S_n^{++}$. \\
$\{E_{ij}\}_{i,j=1}^{n,m}$ are the standard basis vectors in $M_{n,m}(\C)$ with $0$s everywhere except in the $i^{th}$ row and $j^{th}$ column where there is a single $1$. $\{e_i\}_{i=1}^n$ are the standard basis vectors for $\C^n$, considered as $n\times 1$ column vectors.\\
Lastly, we note that we follow the physicists' convention and take our inner products to be anti-linear in the first argument, so that we can easily identify $\langle v,w\rangle = v^*w$, where $v^*$ is the conjugate transpose of the column vector $v$.

\section{Quantum Channels}
\begin{defn} A linear map $\Phi:M_n(\C)\rightarrow M_m(\C)$ is positive if $$X\succeq 0 \Longrightarrow \Phi(X)\succeq 0.$$
\end{defn}

To any linear map acting on a matrix algebra we may define its $k^{th}$ ampliation, the map $(\mathrm{id}\otimes \Phi): M_k(\C)\otimes M_n(\C)\rightarrow M_k(\C)\otimes M_n(\C)$. Essentially, this map acts blockwise on the $k^2$ $n\times n$ blocks of an $nk\times nk$ matrix. 

\begin{defn} If $\Phi:M_n(\C)\rightarrow M_m(\C)$ is a linear map whose $k^{th}$ ampliation is positive for all $k\in \N$, $\Phi$ is said to be completely positive. 
\end{defn}

Completely positive maps were completely characterized by Choi:
\begin{thm}[Choi]\label{choi}\cite{choi}
Given $\Phi:M_n(\C)\rightarrow M_m(\C)$ the following are equivalent:
\begin{enumerate}
\item $\Phi$ is completely positive
\item \label{choimtx}The matrix $M_n(\C)\otimes M_m(\C) \ni C_{\Phi}:=\sum_{i,j=1}^n E_{ij}\otimes \Phi(E_{ij}) \succeq 0$
\item \label{kraus} There exist matrices $\{K_i\}_{i=1}^p$, $K_i \in M_{m,n}(\C)$, that satisfy 
$$\Phi(X) = \sum_{i=1}^p K_i X K_i^*.$$
\end{enumerate}
The matrices in condition \ref{kraus} are called Kraus operators for the channel, and the expression $$\Phi(X) = \sum_{i=1}^p K_i X K_i^*$$ is called an operator sum representation of the channel. The matrix $C_{\Phi}$ is called the Choi matrix, and its existence establishes a one-to-one correspondence between CP maps from $M_n(\C)$ to $M_m(\C)$ and positive semidefinite matrices in $M_n(\C)\otimes M_m(\C)$. 
\end{thm}
Although this result is well-known and can be found in the literature, we will provide a proof for illustrative purposes.

\begin{proof}
The first condition implies the second easily, since $C_{\Phi} = (\mathrm{id}\otimes \Phi)\bigl(\sum_{i,j=1}^n E_{ij}\otimes E_{ij}\bigr)$, and $\sum_{i,j=1}^n E_{ij}\otimes E_{ij} \succeq 0$. Hence $C_{\Phi}\succeq 0$. \\
To see that the second implies the third, suppose $C_{\Phi}\succeq 0$, so that there exist $\{k_i\}_{i=1}^p\in \C^{nm}$ satisfying $$C_{\Phi}= \sum_{i=1}^p k_ik_i^*.$$
Consider the $(i,j)$ block of $C_{\Phi}$, $\Phi(E_{ij})$. If $k_i = \bigoplus_{j=1}^m k_{ij}$ is a division of $k_i$ into $n$ blocks of size $m\times 1$, then this is 
$$\Phi(E_{ij}) = \sum_{l=1}^p k_{li}k_{lj}^*.$$
Now let $K_i$ be the $m\times n$ matrix whose vectorization is $k_i$, that is, $K_i = \sum_{j=1}^n k_{ij}e_j^*$. Then $K_lE_{ij}K_l^* = k_{li}k_{lj}^*$, and so 
$$\Phi(E_{ij}) = \sum_{l=1}^p k_{li}k_{lj}^* = \sum_{l=1}^p K_lE_{ij}K_l^*.$$
That $\Phi(X) = \sum_{l=1}^p K_l X K_l^*$ for general $X$ now follows by linearity.\\
Finally, we show that the third conditions guarantees the first. If $\Phi(X) = \sum_{i=1}^p K_i XK_i^*$, then, for $M=\sum_{i,j=1}^k E_{ij}\otimes M_{ij} \succeq 0$,  
\begin{align*}(\mathrm{id}\otimes \Phi)(M) &= \sum_{i,j=1}^k E_{ij}\otimes \Phi(M_{ij})\\
& = \sum_{i,j=1}^k E_{ij}\otimes \bigl(\sum_{l=1}^p K_lM_{ij}K_l^*\bigr)\\
& = \sum_{l=1}^p(I\otimes K_l)\bigl(\sum_{i,j=1}^k E_{ij}\otimes M_{ij}\bigr)(I\otimes K_l^*)\\
&= \sum_{l=1}^p (I\otimes K_l)M(I\otimes K_l)^*.
\end{align*}

Each term in the sum preserves positivity, and a sum of positive matrices is positive, so the result is positive
\end{proof}

\begin{remark}\label{krausunique} Notice that in the proof of the theorem, the Kraus operators for $\Phi$ correspond to vectors appearing in a resolution of $C_{\Phi}$. Thus, we have as much freedom to pick Kraus operators as we do to pick vectors $\{k_i\}_{i=1}^p$ such that $C_{\Phi} = \sum_{i=1}^p k_ik_i^*$. A canonical choice of $\{k_i\}$ are the scaled eigenvectors of $C_{\Phi}$. Even this choice is not necessarily unique if $C_{\Phi}$ does not have simple spectrum.
\end{remark}

\begin{defn} A quantum channel is a completely positive map that is also trace-preserving.
\end{defn}

\begin{defn} A completely positive map that preserves the identity, i.e., $\Phi(I_n) = I_m$ is said to be unital.
\end{defn}

Both trace-preservation and unitality can be expressed in terms of the Kraus operators: a channel is trace-preserving if 
\begin{align*} \mathrm{Tr}(X)& = \mathrm{Tr}(\Phi(X)) \\
& = \sum_{i=1}^p\mathrm{Tr}(K_iXK_i^*) \\
& =  \mathrm{Tr}(\sum_{i=1}^p K_i^*K_i X)\end{align*} for all $X \in M_n(\C)$; hence $\mathrm{Tr}(X(I_n - \sum_{i=1}^p K_i^*K_i)) = 0$ and so 
$$\sum_{i=1}^p K_i^*K_i = I_n.$$

Unitality corresponds to 
$$\Phi(I_n) = \sum_{i=1}^p K_iK_i^* = I_m.$$

These two notions are dual to each other, in the following sense. We equip the space $M_n(\C)$ with the Hilbert-Schmidt or trace inner product: 
$$\langle A,B\rangle = \mathrm{Tr}(A^*B).$$
Then we can define an adjoint $\Phi^{\dagger}$ (we use $\dagger$ for the adjoint of a linear map between matrix spaces, and $*$ for the adjoint of a matrix to avoid confusion between expressions such as $\Phi(X)^*$ and $\Phi^{\dagger}(X)$)to a CP map via
$$\mathrm{Tr}(\Phi(A)^*B) = \mathrm{Tr}(A^*\Phi^{\dagger}(B)).$$

It is not hard to see that if $\Phi(X) = \sum_{i=1}^p K_iXK_i^*$, then $\Phi^{\dagger}(Y) =\sum_{i=1}^p K_i^*YK_i$; so the adjoint of a completely positive map is completely positive, and if $\Phi$ is trace-preserving, its adjoint is unital, and vice-versa. 

\subsection{Complement Channels}
A famous characterization of quantum channels is via the Stinespring dilation theorem, which we state in the following form:
\begin{thm}[Stinespring]\label{stinespring}\cite{stinespring} If $\Phi:M_n(\C)\rightarrow M_m(\C)$ is a quantum channel, whose Choi matrix has rank $p$, then there exists a unitary $U\in M_{m}(\C)\otimes M_p(\C)$ such that 
$$\Phi(X) = (\mathrm{id}\otimes \mathrm{Tr})\bigl(U(X\otimes E_{11})U^*\bigr).$$
\end{thm}
Again, we will provide a proof for illustrative purposes. 
\begin{proof} Let $V = \sum_{i=1}^p  K_i\otimes e_i$ for $\{K_i\}_{i=1}^p$ a set of Kraus operators for $\Phi$. Since $\Phi$ is trace-preserving we have
$$V^*V = \sum_{i,j=1}^p   K_i^*K_j \otimes e_i^*e_j = \sum_{i=1}^p K_i^*K_i = I_n$$
and so $V$ is an isometry from $\C^n \rightarrow \C^{mp}$. We can complete this to a unitary on $\C^{mp}$: $U = \sum_{i=1,j=1}^p U_{ij} \otimes E_{ij}$, where $U_{i1} = K_i$. The rest is just computation:

\begin{align*} (\mathrm{id}\otimes\mathrm{Tr})\bigl(U(X\otimes E_{11})U^*) & =  \sum_{i,j,k,l=1}^p U_{ij}XU_{kl}^*\mathrm{Tr}(E_{ij}E_{11}E_{lk})\\
&=\sum_{i,k=1}^p U_{i1}XU_{k1}^*\mathrm{Tr}(E_{ik}) \\
& = \sum_{i=1}^p K_i XK_i^*\\
& = \Phi(X).\end{align*} 
\end{proof}

We use the existence of the Stinespring dilation to define the complementary channel to $\Phi$ as follows:
\begin{defn} Let $\Phi:M_n(\C)\rightarrow M_m(\C)$ be a quantum channel with $p$ Kraus operators whose Stinespring representation is 
$$\Phi(X) = (\mathrm{id}\otimes\mathrm{Tr})\bigl(U(X\otimes E_{11})U^*\bigr).$$
Then its complement, or complementary channel, denoted $\Phi^C$, is defined by taking the other partial trace:

$$\Phi^C(X)=(\mathrm{Tr}\otimes \mathrm{id})\bigl( U(X\otimes E_{11})U^*\bigr).$$
It is a map from $M_n(\C)$ to $M_p(\C)$. 
\end{defn}

Given a choice $\{K_i\}_{i=1}^p$ of Kraus operators, we see that 
\begin{align*} \Phi^C(X) & = (\mathrm{Tr}\otimes \mathrm{id})\bigl(U(X\otimes E_{11})U^*\bigr) \\
& = (\mathrm{Tr}\otimes \mathrm{id})\bigl( \sum_{i,j,k,l=1}^p U_{ij}XU_{kl}^*\otimes E_{ij}E_{11}E_{lk}\bigr) \\
& = (\mathrm{Tr}\otimes\mathrm{id})\bigl( \sum_{i,j=1}^p U_{i1}XU_{k1}^*\otimes E_{ik}\bigr)\\
& = \sum_{i=1}^p \mathrm{Tr}(K_k^*K_iX) E_{ik}.\end{align*}

If $\Phi$ is trace-preserving, then so is $\Phi^C$, since the former is equivalent to $\sum_{i=1}^p K_i^*K_i = I_n$, and hence $\mathrm{Tr}(\Phi^C(X)) = \sum_{i=1}^p \mathrm{Tr}(K_i^*K_iX) = \mathrm{Tr}(X)$ . \\
Even when $\Phi$ is not trace-preserving, we will use the expression 
\begin{equation}\label{complement}\Phi^C(X) = \sum_{i,j=1}^p \mathrm{Tr}(K_i^*K_j X)E_{ji}\end{equation} to define a complement for $\Phi$.

\begin{remark} Recall Remark \ref{krausunique}, where we observed that the Kraus operators are not necessarily unique; this obviously affects our expression for $\Phi^C$ using Equation \ref{complement}. We will remove this ambiguity now. Suppose we have two different sets of Kraus operators $\{K_i\}_{i=1}^p$, and $\{L_i\}_{i=1}^q$, corresponding to two different resolutions of $C_{\Phi}$:
$$C_{\Phi} = \sum_{i=1}^p k_ik_i^* = \sum_{i=1}^q l_il_i^*.$$
Without loss of generality, assume $q\geq p$ (we may as well also take $p$ to be the rank of $C_{\Phi}$, and thus the minimum number of Kraus operators).
Let $K =\sum_{i=1}^p k_ie_i^*$, $L=\sum_{i=1}^q l_ie_i^*$, so $KK^* = LL^* = C_{\Phi}$. 
Then $L=KU^*$ for $U$ satisfying $U^*U=I$, and so $l_i = \sum_{j=1}^p \overline{u_{ij}}k_j$, and so also $L_i = \sum_{j=1}^p \overline{u_{ij}}K_j$.
\end{remark}
The above remark tells us how $\Phi^C$ transforms by taking taking different Kraus operators to represent $\Phi$. Let $\{K_i\}_{i=1}^p$, $\{L_i\}_{i=1}^q$, and $U$ all be as in the preceding remark. Then
\begin{align*} \overline{U}\Phi^C(X)U^T & = \sum_{i,j=1}^{q,p}\sum_{k,l=1}^p\sum_{r,s=1}^{p,q} \overline{u_{ij}}u_{sr}\mathrm{Tr}(K_l^*K_kX) E_{ij}E_{kl} E_{rs} \\
& = \sum_{i,s=1}^qE_{is} \bigl(\sum_{j,r=1}^p \overline{u_{ij}}u_{sr}\mathrm{Tr}(K_r^*K_jX)\bigr)\\
& = \sum_{i,s=1}^qE_{is} \mathrm{Tr}\biggl((\sum_{r=1}^p u_{sr}K_r^*)(\sum_{j=1}^p \overline{u_{ij}}K_j)X\biggr) \\
& = \sum_{i,r}E_{ir} \mathrm{Tr}(L_r^*L_iX)\end{align*} which is the expression for $\Phi^C$ we obtain if we choose $\{L_i\}$ as our set of Kraus operators.\\
Since $U^T\overline{U} = (\overline{U^*U})= I$, Equation \ref{complement} defines the complementary channel of $\Phi$ up to adjunction by an isometry. $\mathrm{Tr}(U\Phi^C(X)U^*) = \mathrm{Tr}(\Phi^C(X))$, so this is still trace-preserving.
Keeping this in mind, we also find an expression for the adjoint of the complement: 

$$\mathrm{Tr}(\Phi^{C\dagger}(E_{ij})X) = \mathrm{Tr}(E_{ij}\Phi^C(X)) = \mathrm{Tr}(K_i^*K_jX)$$
for all $X$, so \begin{equation}\label{phicdag}\Phi^{C\dagger}(E_{ij}) = K_i^*K_j.\end{equation} 
This means that $\Phi^{C\dagger}(I) = \sum_{i=1}^p K_i^*K_i = I_n$ as long as $\Phi$ is trace-preserving; hence, the range of $\Phi^{C\dagger}$ is $*$-closed and contains the identity. We will come back to this point later. 

\subsection{Schur Product Channels}
\begin{defn} Given two matrices $A,B\in M_{n,m}(\C)$, their Schur product $A\circ B$ is the product of their individual entries:
$$(A\circ B)_{ij} = a_{ij}b_{ij}.$$
\end{defn}

It is easy to see that, for vectors $v,w$, $X\circ (vw^T) = D_vXD_w$ where $D_v$ denotes the diagonal matrix with $v$ down its diagonal. Hence, if $C\succeq 0$, so that $C = \sum_{i=1}^p v_iv_i^*$, $X\circ C = \sum_{i=1}^p D_{v_i}XD_{v_i}^*$. This is an operator-sum form for the map $\Phi(X)=X\circ C$, and so Schur product with some fixed $C\succeq 0$ is completely positive. Conversely, any $CP$ map with diagonal Kraus operators has this form. Such a map is called a Schur product map. \\
A Schur product map is trace-preserving if and only if 
$$\sum_{i=1}^n x_{ii}c_{ii}= \sum_{i=1}^p x_{ii}$$ for all $X\in M_n(\C)$, in which case $c_{ii}=1$ for all $1\leq i \leq n$. Thus, a trace-preserving Schur product map, or Schur product channel, has the form 
\begin{equation}\Phi(X)=X\circ C \end{equation} for some $C\succeq 0$ with $1$s down the diagonal.

\begin{defn} A positive-semidefinite matrix with $1$s down the diagonal is called a correlation matrix. We denote the set of $n\times n$ correlation matrices by $\mathcal{E}_n$, and call this set the $n$-dimensional elliptope. 
\end{defn}
A Schur product channel is automatically a unital channel as well, since if $c_{ii}=1$, $\Phi(X) = X\circ C$ satisfies $\Phi(I) = I\circ C = \mathrm{diag}(C)=I$.

\begin{remark}\label{eliptope} The set $\mathcal{E}_n$ is a closed convex set whose structure we will discuss more later; for now we note that $\mathcal{E}_n$ is the intersection of the convex cone of positive semidefinite $n\times n$ matrices with the linear subspace defined by $c_{ii}=1$, and so any extreme point of $S_{n}^+$ lying in this subspace is also an extreme point of $\mathcal{E}_n$. Since the extreme points of $S_{n}^+$ are rank-one matrices, any correlation matrix of the form $zz^*$ is an extreme point of $\mathcal{E}_n$. 
\end{remark}

Next, we discuss the structure of $\Phi^C$ and $\Phi^{C\dagger}$ for a Schur product channel $\Phi$.

\begin{lem} Let $C$ be an $n\times n$ correlation matrix of rank $p$. There exist unit vectors $\{w_i\}_{i=1}^n\subseteq \C^p$, with $\langle w_i,w_j\rangle = c_{ij}$ and any such set $\{w_i\}_{i=1}^n$ corresponds to a resolution for $C$ of the form $C=\sum_{i=1}^p v_iv_i^*$ and vice-versa by means of the identification $v_{ij} = \overline{w_{ji}}$ where $v_{ij}$ is the $j^{th}$ entry of the vector $v_i$ and $w_{ji}$ is the $i^{th}$ entry of the vector $w_j$.
\end{lem}
\begin{proof} Let $v_i = \sum_{j=1}^n v_{ij}e_j$ be vectors in a resolution for $C$, so that $C = \sum_{i=1}^k v_iv_i^*$. We can without loss of generality take $k=p$. Define $w_i$ by $w_i = \sum_{j=1}^k \overline{v_{ji}}e_j$, so that if $V = \sum_{i=1}^k v_ie_i^*$ has $v_i$ as its columns, it has $w_i^*$ as its rows. Then $C = VV^*$, and so $c_{ij} = w_i^*w_j$. Since $c_{ii} = w_i^*w_i = \|w_i\|^2=1$ each $w_i$ must be a unit vector.\\
The converse proceeds in the opposition direction: if $w_i^*w_j = c_{ij}$, let $V$ be the matrix whose rows are $w_i^*$; then $C = VV^*$. Then the columns of $V$, $v_i$, satisfy $C = \sum_{i=1}^k v_iv_i^*$. 
\end{proof}

\begin{prop} Let $C\in \mathcal{E}_n$, and $\Phi(X) = X\circ C$. Let $\{w_i\}_{i=1}^n$ in $\C^p$ be a set of unit vectors satisfying $w_i^*w_j = c_{ij}$. Let $\{v_i\}_{i=1}^p$ in $\C^n$ be the corresponding resolution of $C$, $C = \sum_{i=1}^p v_iv_i^*$. Then 
\begin{equation} \Phi^C(X) = \sum_{i=1}^n x_{ii}(w_iw_i^*)^T \end{equation}
and 
\begin{equation}\Phi^{C\dagger}(X) = \sum_{i=1}^n (w_i^TX\overline{w_i})E_{ii}.\end{equation}
\end{prop}
\begin{proof} Choose $\{D_{v_i}\}_{i=1}^p$ to be a set of Kraus operators for the channel $\Phi$. Then $\Phi^C(X) = \sum_{i,j=1}^p \mathrm{Tr}(D_{v_j}^*D_{v_i}X)E_{ij}$. If $X=E_{kl}$ for $k\neq l$, all terms vanish since $D_{v_j}^*D_{v_i}$ is diagonal. If $X=E_{kk}$, we get
\begin{align*} \sum_{i,j=1}^p \mathrm{Tr}(D_{v_j}^*D_{v_i}E_{kk})E_{ij} & = \sum_{i,j=1}^p \overline{v_{jk}}v_{ik} E_{ij} \\
& = \sum_{i,j=1}^p w_{kj}\overline{w_{ki}}E_{ij}\\
& = \bigl(\sum_{i=1}^p \overline{w_{ki}}e_i\bigr)\bigl(\sum_{j=1}^p w_{kj}e_j^*\bigr)\\
& = \overline{w_k}w_k^T\\
& = (w_kw_k^*)^T.\end{align*}

From this, it's easy to see that a choice of Kraus operators for $\Phi^C$ is $\{W_i\}_{i=1}^n$, where $W_i = \overline{w_i}e_i^*$. Then $\{W_i^*\}_{i=1}^n$ is a set of Kraus operators for $\Phi^{C\dagger}$, and so 
\begin{align*} \Phi^{C\dagger}(X)& = \sum_{i=1}^n W_i^*XW_i \\
& = \sum_{i=1}^n (e_iw_i^T)X(\overline{w_i}e_i^*)\\
& = \sum_{i=1}^n (w_i^TXw_i)E_{ii}.\end{align*}
\end{proof}

\section{Factorizable Channels}
\subsection{Factorizations}
Factorizations of quantum channels have been studied in both quantum information and purely for their own mathematical interest. Factorizations have been considered in the context of dilations \cite{anantharaman}, and in the study of the so-called Asymptotic Quantum Birkhoff Conjecture \cite{haagerup}. In \cite{haagerup}, Haagerup and Musat found a number of equivalent conditions for a channel to admit a factorization; the definition we give below is just one of their equivalent conditions, and not the one in terms of which factorizations were initially defined, but it is the one most useful for our purposes.

\begin{defn}\label{factdefn} Let $\Phi:M_n(\C)\rightarrow M_n(\C)$ be a quantum channel with Kraus operators $\{K_i\}_{i=1}^p$. $\Phi$ is said to be factorizable if either of the following equivalent conditions hold:
\begin{enumerate}
\item There exists a finite, tracial von Neumann algebra $\mathcal{N}$, and $\{V_i\}_{i=1}^p$, $V_i \in \mathcal{N}$, $\mathrm{Tr}_{\mathcal{N}}(V_i^*V_j)=\delta_{ij}$ such that
$$U = \sum_{i=1}^p K_i\otimes V_i \in M_{n}(\C)\otimes \mathcal{N}$$ is unitary.
\item There exists a finite tracial von Neumann algebra $\mathcal{N}$ and a unitary $U \in M_n(\C)\otimes \mathcal{N}$ such that
$$\Phi(X) = (\mathrm{id}\otimes \mathrm{Tr}_{\mathcal{N}})\bigl(U(X\otimes I_{\mathcal{N}})U^*\bigr).$$
\end{enumerate}
We will often take a short cut and refer to the algebra $\mathcal{N}$ itself, or even a particular choice of $\{V_i\}_{i=1}^p$ as a factorization for $\Phi$. 
A channel is matrix factorizable if the algebra $\mathcal{N}$ is a matrix algebra; i.e., if there exist positive integers $\{i_k\}_{k=1}^M$ such that $\mathcal{N} = \bigoplus_{k=1}^M M_{i_k}(\C)$.
\end{defn}
The equivalence to the two conditions in the above Definition is proved in Theorem $2.2$ of \cite{haagerup}, along with their equivalence to a third condition from which factorizability takes its name; but it is the two above conditions that are most useful to us.\\ 
We begin by noting another equivalent condition for the existence of a factorization.
\begin{prop}\label{factbycomplement} Let $\Phi:M_n(\C)\rightarrow M_n(\C)$ be a quantum channel. Then $\Phi$ is factorizable if and only if there exists a finite, tracial von Neumann algebra $\mathcal{N}$, and $\{V_i\}_{i=1}^p$, $V_i \in \mathcal{N}$, $\mathrm{Tr}_{\mathcal{N}}(V_i^*V_j)=\delta_{ij}$ such that for all $A \in \mathrm{range}(\Phi^C)$, 
$$\sum_{i,j=1}^p a_{ij}V_j^*V_i = \mathrm{Tr}(A)I_{\mathcal{N}}.$$
\end{prop}
\begin{proof}
First, suppose $\Phi$ is factorizable, with $\{V_i\}_{i=1}^p$ forming a factorization. If $\{K_i\}_{i=1}^p$ is a set of Kraus operators for $\Phi$, we must have that  
$$U^*U = \sum_{i,j=1}^p K_i^*K_j \otimes V_i^*V_j = I_n\otimes I_{\mathcal{N}}.$$

$K_i^*K_j \otimes V_i^*V_j = \sum_{k,l=1}^n (K_i^*K_j)_{kl}E_{kl}\otimes V_i^*V_j$, and so we must have
$$\sum_{k,l=1}^n E_{kl}\otimes\bigl(\sum_{i,j=1}^p (K_i^*K_j)_{kl}V_i^*V_j\bigr) = I_n\otimes I_{\mathcal{N}}$$ from which we see that
$\sum_{i,j=1}^p (K_i^*K_j)_{kl}V_i^*V_j = \delta_{kl}I_{\mathcal{N}}$.
Next, recall from Equation \ref{phicdag} that $K_i^*K_j = \Phi^{C\dagger}(E_{ij})$, and so 
$$(K_i^*K_j)_{kl}= \mathrm{Tr}(\Phi^{C\dagger}(E_{ij})E_{lk}) = \mathrm{Tr}(E_{ij}\Phi^C(E_{lk})) = \Phi^C(E_{lk})_{ji}.$$

Hence $$\sum_{i,j=1}^p \Phi^C(E_{kl})_{ji}V_i^*V_j = \delta_{kl}I_{\mathcal{N}} = \mathrm{Tr}(E_{kl})I_{\mathcal{N}}$$ and extending by linearity we have one direction.

Conversely, suppose $\{V_i\}_{i=1}^p$ with the stated properties exist. Then, applying the condition to each of $\Phi^C(E_{kl})$ in turn, and reversing the steps above, we see that if $U = \sum_{i=1}^p K_i\otimes V_i$, $U^*U = I_n\otimes I_{\mathcal{N}}$. 
\end{proof}

\begin{remark}From Proposition \ref{factbycomplement} we infer that factorizability of a channel depends only on the structure of the range of its complement; we have already noted that this is only defined up to adjunction by an isometry.\\
We can equally well phrase this observation in terms of the kernel of $\Phi^{C\dagger}$, since the range of $\Phi^C$ and the kernel of $\Phi^{C\dagger}$ are orthogonal complements in $M_p(\C)$: factorizability of $\Phi$ depends only on the structure of $\mathrm{kernel}(\Phi^{C\dagger})$.
\end{remark}

\begin{prop}\label{factisconv} Let $\Phi_1$, $\Phi_2$ be two factorizable quantum channels, represented by sets of Kraus operators $\{K_i\}_{i=1}^p$ and $\{L_i\}_{i=1}^q$. Suppose that $\mathcal{N}_1$, $\mathcal{N}_2$ are algebras that act as factorizations for the respective channels.\\
Then the convex combination $\Phi=t\Phi_1 + (1-t)\Phi_2$ for $t\in \left(0,1\right)$ is factorizable by means of the algebra $\mathcal{N}=\mathcal{N}_1\oplus\mathcal{N}_2$ with trace $\mathrm{Tr}_{\mathcal{N}}(A\oplus B) = t\mathrm{Tr}_{\mathcal{N}_1}(A)+(1-t)\mathrm{Tr}_{\mathcal{N}_2}(B)$.
\end{prop}
\begin{proof} Suppose $\{\widehat{V_i}\}_{i=1}^p$ and $\{\widehat{W_i}\}_{i=1}^q$ are operators in $\mathcal{N}_1$, $\mathcal{N}_2$ by which the two channels are respectively factorized. Let $\{X_i\}_{i=1}^{p+q}$ be given by $X_i = \left[(\sqrt{t}^{-1}\widehat{V_i})\oplus 0\right] \in \mathcal{N}$ if $1\leq i \leq p$, and $X_i = \left[0\oplus (\sqrt{1-t}^{-1}\widehat{W_{i-p}})\right]$ if $p+1\leq i \leq p+q$. Notice that $X_i^*X_j = 0$ if $1\leq i\leq p$ and $p+1\leq j \leq p+q$ or vice-versa, so such pairs are automatically orthogonal. Otherwise, $\mathrm{Tr}_{\mathcal{N}}(X_i^*X_j) = t\mathrm{Tr}(t^{-1}\widehat{V_i}^*\widehat{V_j}) = \delta_{ij}$ or $(t-1)\mathrm{Tr}((t-1)^{-1}\widehat{W_{i-p}}^*\widehat{W_{j-p}}) = \delta_{ij}$. \\
Clearly, a set of Kraus operators for $\Phi$ is $\{\sqrt{t}K_i\}_{i=1}^p \cup \{\sqrt{1-t}L_i\}_{i=1}^q$. Then \begin{align*}\Phi^C(X) &= t\sum_{i,j=1}^p \mathrm{Tr}(K_j^*K_iX)E_{ij} + \sqrt{t(1-t)}\sum_{i,j=1}^{p,q}\mathrm{Tr}(K_i^*L_jX)E_{p+j,i}\\
& + \sqrt{t(1-t)}\sum_{i,j=1}^{q,p} \mathrm{Tr}(L_i^*K_jX)E_{j,p+i} + (1-t)\sum_{i,j=1}^{q}\mathrm{Tr}(L_i^*L_j X)E_{p+j,p+i}.\end{align*}

For any $X$, take $\sum_{i,j=1}^{p+q}\Phi^C(X)_{ji}X_i^*X_j$. We have already noticed that terms where $i,j$ are not both between either $1$ and $p$ or $p+1$ and $p+q$ disappear, so we get only
\begin{align*}&t \sum_{i,j=1}^p \mathrm{Tr}(K_j^*K_iX)\bigl[(t^{-1}\widehat{V_j}^*\widehat{V_i})\oplus 0\bigr] + (1-t)\sum_{i,j=1}^q \mathrm{Tr}(L_j^*L_iX)\bigl[ 0\oplus((1-t)^{-1}\widehat{W_j}^*\widehat{W_i})\bigr]  =\\
& \bigl(\sum_{i,j=1}^p \Phi_1^C(X)_{ij}\widehat{V_j}^*\widehat{V_i}\bigr)\oplus\bigl(\sum_{i,j=1}^q \Phi_2^C(X)_{ij} \widehat{W_j}^*\widehat{W_i}\bigr) \\
& = I_{\mathcal{N}_1}\oplus I_{\mathcal{N}_2} \\
& = I_{\mathcal{N}}.\end{align*}

\end{proof}
\begin{remark} The above proposition proves not only that the set of factorizable quantum channels is convex, but that the among the factorizations of a channel $\Phi$ that is a convex combination of other factorizable channels $\{\Phi_i\}$, are factorizations it inherits from the factorizations associated with the $\{\Phi_i\}$, with convex combination of channels becoming direct sum of factorizations. Our main theorem is a converse to this proposition.
\end{remark}
  
\section{Factorizations and Linear Matrix Inequalities}
\begin{defn} Let $Z = (Z_1,\cdots,Z_d)\in S_n^d$. The linear matrix inequality defined by $Z$, $L_Z(A)$, is the family of inequalities for $A = (A_1,A_2,\cdots, A_d)$
\begin{equation}\label{LMI} L_Z(A) := I\otimes I_N + \sum_{i=1}^d Z_i\otimes A_i\succeq 0\end{equation} for $A_i \in M_N(\C)$. 
\end{defn}

In a sense, a linear matrix inequality, or LMI, is an inequality where we allow matrix solutions of arbitrary size; we do not assume we know the preferred dimension of our solution matrices ahead of time. We group solutions into sets by size: 
\begin{equation}\label{Dk}\mathcal{D}_Z(k):=\{A \in S_k^d: L_Z(A)\succeq 0\}\end{equation} and then collect all the solutions up together in the set 
\begin{equation}\label{D}\mathcal{D}_Z:=\bigcup_{k=1}^{\infty} \mathcal{D}_Z(k).\end{equation}

The set $\mathcal{D}_Z(1) = \{x\in \R^d: L_Z(x) \succeq 0\}$ is sometimes called a spectrahedron, and Helton, Kelp, and McCullough \cite{klep} call the set $\mathcal{D}_Z$ the free spectrahedron. Solutions sets of a linear matrix inequality form a matrix convex set: a set 
$$\mathcal{K} = \cup_{i=1}^{\infty} \mathcal{K}(i)$$ where each $\mathcal{K}(i) \subseteq S_i^d$ satisfies the following conditions:
\begin{enumerate}
\item\label{graded} It is closed under direct sums: if $(A_1,\cdots, A_d) = A \in \mathcal{K}(i)$ and $(B_1,\cdots,B_d) = B\in \mathcal{K}(j)$, then $$A\oplus B :=(A_1\oplus B_1,\cdots, A_d \oplus B_d) \in \mathcal{K}(i+j).$$
\item\label{free} It is closed under unitary equivalence: if $A \in \mathcal{K}(i)$, and $U \in M_i(\C)$ is unitary, then $$UAU^* = (UA_1U^*,\cdots, UA_dU^*) \in \mathcal{K}(i).$$
\item\label{convex} It is closed under isometric adjunction: if $A \in \mathcal{K}(i)$, and $V: \C^j\rightarrow \C^i$ is an isometry, then 
$$V^*AV = (V^*A_1V, \cdots, V^*A_dV) \in \mathcal{K}(j).$$
\end{enumerate}
A set satisfying satisfying just the first is a graded set; just the second is free. The first and third properties alone are enough to characterize matrix convex sets. It is easy to verify that the solutions to an LMI form a matrix convex set. The set of solutions to an LMI also satisfies an extra property: if $A \in \mathcal{D}_Z(i+j)$ is reducible, so that the $A_i$ can be simultaneously reduced to $A_i \simeq B_i \oplus C_i$, for $B \in S_i^d$ and $C\in S_j^d$, then $A \in \mathcal{D}_Z(i)$ and $B \in \mathcal{D}_Z(j)$. \\

Linear matrix inequalities have been studied in the context of quantum information in particular because of their use in determining when certain interpolation problems can be satisfied by means of unital CP maps \cite{davidson}.  

In this section, we will show that finding matrix factorizations for a quantum channel $\Phi$ is equivalent to finding certain solutions to an LMI defined by $\Phi$.

\begin{thm} Let $\{Z_i\}_{i=1}^d$ be a self-adjoint basis for the orthogonal complement to the range of $\Phi^C$. Then there exists $k$ and $\{V_i\}_{i=1}^p \in M_k(\C)$ such that $U = \sum_{i=1}^p K_i\otimes V_i$ is unitary if and only if there exists $A \in S_k^d$ such that
$$L_Z(A) \succeq 0$$ and $\mathrm{rank}(L_Z(A)) \leq k$. 
The matrices $\{V_i\}$ are orthonormal with respect to the Hilbert-Schmidt inner product if and only if $$\R^d \ni(\mathrm{Tr}(A_1),\cdots,\mathrm{Tr}(A_d)) =:\mathrm{Tr}(A)=0.$$
\end{thm}
\begin{proof} Suppose such an $A$ exists: then $L_Z(A) = I_p\otimes I_k + \sum_{i=1}^d Z_i \otimes A_i \succeq 0$ can be factored as $L_Z(A) = V^*V$. The rank of $L_Z(A)$ is $r\leq k$, so we can choose $V$ to have $k$ rows by appending $k-r$ rows of $0$s to $V$. Divide $V$ into $p$ blocks of size $k\times k$, which we call $V_i$, so that $V = \sum_i e_i^T \otimes V_i$; then $V^*V = \sum_{i,j} E_{ij}\otimes V_i^*V_j$; that is, it is a block matrix whose blocks are $V_i^*V_j$ for $V_i \in M_k(\C)$. \\
Let $X \in \mathrm{range}(\Phi^C)$, so that $X \perp Z_i$ for each $i$.\\
Then 
$$(\mathrm{Tr}\otimes \mathrm{id})\bigl( (X\otimes I_k)L_Z(A) \bigr) = \mathrm{Tr}(X)I_k + \sum_{i=1}^d \mathrm{Tr}(XZ_i)A_i  = \mathrm{Tr}(X)I_k$$
but recalling that $L_Z(A) = \sum_{i,j}E_{ij}\otimes V_i^*V_j$ we see that 
$$(\mathrm{Tr}\otimes \mathrm{id})\bigl( (X\otimes I_k)L_Z(A)\bigr) = \sum_{i,j} X_{ji} V_i^*V_j = \mathrm{Tr}(X)I_k$$ and so the matrices $\{V_i\}$ satisfy Proposition \ref{factbycomplement}, and $U = \sum_{i=1}^p K_i \otimes V_i$ is unitary.\\
For the converse, let $\{V_i\}_{i=1}^p$ be $k\times k$ matrices satisfying $U=\sum_{i=1}^p K_i \otimes V_i$ is unitary. Let $V^*V =  E_{ij}\otimes V_i^*V_j$; clearly this matrix is positive semidefinite, has rank less than or equal to $k$, and by Proposition \ref{factbycomplement} it must satisfy
$$(\mathrm{Tr}\otimes \mathrm{id})\bigl((X\otimes I)V^*V\bigr) = \mathrm{Tr}(X)I_k$$ for each $X \perp \ker (\Phi^{C\dagger})$. If we express $V^*V = \sum_i B_i \otimes A_i \in M_p(\C)\otimes M_k(\C)$, we see that we can choose $B_1 = I_p$, $A_1 = I_k$; then $\mathrm{Tr}(XB_i) = 0$ for each $B_i$, and $B_i \in \ker (\Phi^{C\dagger})$, so we can take $B_i = Z_i$. Hence $V^*V = L_Z(A)$ for some $A$, and so $A \in \mathcal{D}_Z(k)$.\\
Finally, $\{V_i\}$ are trace-orthonormal if and only if $(\mathrm{id}\otimes \mathrm{Tr})(V^*V) = I_p$, so that $I_p + \sum_{i=1}^d \mathrm{Tr}(A_i)Z_i = I_p$. This happens if and only if $\mathrm{Tr}(A_i) = 0$ for each $i$. 
\end{proof}

\begin{remark} The Theorem and proof above work replacing $M_k(\C)$ with an arbitrary von Neumann algebra $\mathcal{N}$: $\{V_i\}_{i=1}^p$ in $\mathcal{N}$ satisfy 
$U = \sum_{i=1}^p K_i \otimes V_i$ is unitary if and only if there exists $A \in \mathcal{N}^d$ such that $L_Z(A) \succeq 0$ and $L_Z(A)$ can be factored as $\sum_{i,j=1}^p E_{ij}\otimes V_i^*V_j$ for $V_i \in \mathcal{N}$.  
\end{remark}

We will end this section by discussing the solutions $\mathcal{D}_Z(1)$, which have a special interpretation. \\
We have already noted that if $\{Z_i\}_{i=1}^d$ are a basis for the orthogonal complement of $\mathrm{range}(\Phi^C)$, they are also a basis for $\mathrm{ker}(\Phi^{C\dagger})$. Recall that $\Phi^{C\dagger}$ is unital so long as $\Phi$ is trace-preserving; it is also, by Equation \ref{phicdag}, easy to see that $\Phi^{C\dagger}(A)^* = \Phi^{C\dagger}(A^*)$, so if $\Phi^{C\dagger}(A)=0$, the same is true for $A^*$. For this reason, we can always pick a basis for $\mathrm{kernel}(\Phi^{C\dagger})$ to be composed only of self-adjoint matrices. The next result is essentially a restatement of a theorem of Choi characterizing when a map $\Phi$ can be written as a convex combination of other completely positive maps.

\begin{thm}[Choi]\cite{choi}\label{choixtmpts} Let $\Phi$ be a trace-preserving completely positive map with Kraus operators $\{K_i\}_{i=1}^p$, and $\{Z_i\}_{i=1}^d$ a self-adjoint basis for $\mathrm{kernel}(\Phi^{C\dagger})$. Also let $K$ be the matrix whose columns are $k_i$, the vectorizations of the Kraus operators $k_i$, so that $C_{\Phi} = KK^*$.\\
Then $\mathcal{D}_Z(1)$ parametrizes the minimal face $\mathcal{F}$ of the set of trace-preserving completely positive maps that contains $\Phi$ as follows: for each $x\in \mathcal{D}_Z(1)$, and $Q^*Q = L_Z(x)$, the CP map whose Choi matrix is 
$$K^*(Q^*Q)^TK$$ is a trace-preserving CP map in the same face as $\Phi$.
\end{thm}
\begin{proof} As usual, we will provide our own proof, modified from Choi's original, as it is instructive for our purposes. Suppose $x \in \mathcal{D}_Z(1)\subseteq \R^d$, so $L_Z(x)=I_p + \sum_{i=1}^d x_i Z_i \succeq 0$. Then there exists $Q$ such that $Q^*Q = L_Z(x)$, and $$\Phi^{C\dagger}(Q^*Q) = \Phi^{C\dagger}(I_p) + \sum_{i=1}^d x_i \Phi^{C\dagger}(Z_i) = I_p.$$

By Equation \ref{phicdag}, $$\Phi^{C\dagger}(Q^*Q) = \sum_{i,j}(Q^*Q)_{ij}K_i^*K_j = \sum_{i,j=1}^p \sum_{k=1}^r \overline{q_{ki}}q_{kj}K_i^*K_j = \sum_{k=1}^r \bigl(\sum_{i=1}^p q_{ki}K_i\bigr)^*\bigl(\sum_{j=1}^r q_{kj}K_j\bigr) = I.$$

Let $\widehat{K_k} = \sum_{j=1}^r q_{kj}K_j$, and let $\widehat{\Phi}$ be the channel whose Kraus operators are $\{\widehat{K_k}\}_{k=1}^r$. The above calculation establishes that $\widehat{\Phi}$ is trace-preserving. We now show it lies in $\mathcal{F}$, the same face as $\Phi$. \\
To do so, we observe that if $\Phi$ is a convex combination of channels $\Phi=t\Phi_1+(1-t)\Phi_2$, then the same is true for the associated Choi matrices: $C_{\Phi}=tC_{\Phi_1}+(1-t)C_{\Phi_2}$, and so we just need to show that $C_{\Phi}$ and $C_{\widehat{\Phi}}$ lie in the same face of the convex subset of $S_{nm}^+$ satisfying $(\mathrm{id}\otimes\mathrm{Tr})(A)= I_n$. But the facial structure of this convex set is inherited from the facial structure of $S_{nm}^+$, where $A$ lies in the minimal face containing $B$ if and only if $\mathrm{kernel}(A)\subseteq \mathrm{kernel}(B)$\cite{}. \\
Since $C_{\widehat{\Phi}} = \sum_{i=1}^p \widehat{k_i}\widehat{k_i}^*=\widehat{K}\widehat{K}^*$, and since $\widehat{k_i} = \sum_{j=1}^r q_{ij}k_i$, we see that 
$$C_{\widehat{\Phi}}=\widehat{K}\widehat{K}^* = (KQ^T)(\overline{Q}K^*).$$
Suppose now $x\in \mathrm{kernel}(C_{\Phi})$; then $KK^*x = 0$, and so $K^*x=0$. But then $C_{\widehat{\Phi}}x = K(Q^T\overline{Q})K^*x = 0$ as well, and so $\mathrm{kernel}(C_{\Phi})\subseteq \mathrm{kernel}(C_{\widehat{\Phi}})$, and so $\widehat{\Phi}$ is contained in $\mathcal{F}$.\\
To prove the converse, we just reverse the steps: the face of completely positive maps containing $\Phi$ is the set of positive matrices of the form $K(Q^*Q)^TK^*$; in order to preserve the fact that $\Phi$ is trace-preserving, we need to restrict to $X$ such that $K(Q^*Q)^TK^*$ is the Choi matrix for a trace-preserving channel. The Kraus operators for this channel are $L_i=\sum_{j=1}^r q_{ji}K_i$, so we need that 
$$\sum_{i=1}^r L_i^*L_i = \sum_{i=1}^r \sum_{j,k=1}^p \overline{q_{ij}}q_{ik}K_j^*K_k = \Phi^{C\dagger}(Q^*Q)=I_n.$$

Since $\Phi^{C\dagger}(Q^*Q)= \Phi^{C\dagger}(I_p)=I_n$, $I_p -Q^*Q \in \mathrm{kernel}(\Phi^{C\dagger})$, so $Q^*Q = I_p + \sum_{i=1}^d x_i Z_i$ for some $x\in \R^d$, and clearly $Q^*Q \succeq 0$, so $x\in \mathcal{D}_Z(1)$. 
\end{proof}

\begin{remark} In the case where $\Phi$ is a Schur product channel $\Phi(X)=X\circ C$, the face $\mathcal{F}$ corresponds to the minimal face of $\mathcal{E}_n$ containing $C$. In this case, the result above is, essentially, a result of Li and Tam \cite{li} characterizing the facial structure of $\mathcal{E}_n$. 
\end{remark}

\begin{cor}[Choi]\cite{choi}\label{xtmmaps} A quantum channel $\Phi$ is an extreme point in the set of quantum channels if and only if the map $\Phi^{C\dagger}$ is invertible. 
\end{cor}

\begin{remark}\label{fibresdescend} For $A \in \mathcal{D}_Z(k)$, $(\mathrm{id}\otimes \mathrm{Tr})(L_Z(A))= Q^*Q$, where $\Phi^{C\dagger}(Q^*Q) = I$. To see this, recall that partial trace is completely positive, hence $(\mathrm{id}\otimes\mathrm{Tr})(L_Z(A)) = I_p + \sum_{i=1}^d \mathrm{Tr}(A_i)Z_i \succeq 0$, so $a:=\mathrm{Tr}(A) \in \mathcal{D}_Z(1)$, and $L_Z(a)$ is the image of this point, and hence satisfies the conditions of Theorem \ref{choixtmpts}. 
\end{remark}
%
%

\section{Convex Combinations and Matrix Factorizations}
We have already seen that if $\Phi$ is the convex combination of $\Phi_1$ and $\Phi_2$, which are factorizable by means of $\mathcal{N}_1$ and $\mathcal{N}_2$, then $\Phi$ is factorizable by means of the direct sum of $\mathcal{N}_i$. Our main result in this section is to establish a converse.

\begin{lem}\label{imageunderhomos} Suppose $\Phi$ is a quantum channel with Kraus operators $\{K_i\}_{i=1}^p$ and is factorized by the algebra $\mathcal{N}$, by means of $\{V_i\}_{i=1}^p$ in $\mathcal{N}$, so that $U = \sum_{i=1}^p K_i \otimes V_i$ is unitary. Let $\Psi:\mathcal{N}\rightarrow \mathcal{M}$ be a unital $*$-homomorphism; then $W = \sum_{i=1}^p K_i \otimes \Psi(V_i)$ is unitary as well, and yields a factorization for some channel in the same face as $\Phi$.
\end{lem}
\begin{proof} Let $X \in \mathrm{range} (\Phi^C)$. Then 
\begin{align*} \sum_{i,j} X_{ji} \Psi(V_i^*)\Psi(V_j) & = \sum_{i,j} X_{ji}\Psi(V_i^*V_j) \\
& = \Psi(\sum_{i,j} X_{ji}V_i^*V_j) \\
& = \Psi(I_{\mathcal{N}}) \\
& = I_{\mathcal{M}}.
\end{align*}
So $W$ is unitary, and therefore $\sum_{i,j}E_{ij}\otimes \Psi(V_i^*)\Psi(V_j) = L_Z(A)$ for some $A \in \mathcal{M}^d$.\\
Let $\Phi_k(Y) = (\mathrm{id}\otimes \mathrm{Tr})\bigl(W( Y\otimes I_{\mathcal{N}})W^*\bigr)$. If $Q^*Q = (\mathrm{id}\otimes \mathrm{Tr})(L_Z(A))$ then 
\begin{align*}\Phi_k(Y) &= \sum_{i,j}(Q^*Q)_{ij} K_iYK_j^*\\
& = \sum_{i,j} \sum_k \overline{q_{ki}}q_{kj}K_iYK_j^* \\
& = \sum_k (\sum_i \overline{q_{ki}}K_i)Y(\sum_j q_{kj}K_j^*)\\
& = \sum_k \widehat{K_k}Y\widehat{K_k}^*.\end{align*} 
By Remark \ref{fibresdescend}, $\Phi^{C\dagger}(Q^*Q) = I$, and so Theorem \ref{choixtmpts} guarantees that $\Phi_k$ is in the same face as $\Phi$. 
\end{proof}
\begin{remark} Notice that if $\Psi$ is an isometry, then $W$ is a factorization for $\Phi$ itself. Only in the case that $\Psi$ is a $*$-homomorphism, but is not isometric, do we get factorizations for other channels in the same face as $\Phi$. 
\end{remark}


%

We are now in a position to prove the converse to Proposition \ref{factisconv}.
\begin{thm}\label{main} Let $\Phi$ be a quantum channel with Kraus operators $\{K_i\}_{i=1}^p$. Suppose that $\mathcal{N} \simeq \bigoplus_{k=1}^M \mathcal{N}_k$ with trace $\mathrm{Tr}_{\mathcal{N}}(\oplus_{k=1}^M A_k) = \sum_{k=1}^M q_k \mathrm{Tr}_{\mathcal{N}_k}(A_k)$ is a factorization for $\Phi$, by means of the unitary $U = \sum_{i=1}^p K_i\otimes V_i$, where $\{V_i\}_{i=1}^p$ are trace-orthonormal in $\mathcal{N}$. Then there exist quantum channels $\Phi_k$ such that $\Phi = \sum_{k=1}^M q_k \Phi_k$, and each $\Phi_k$ is factorizable by $\mathcal{N}_k$. 
\end{thm}
\begin{proof} $\mathcal{N}$ admits $*$-homomorphisms $\Psi_k$ onto each of its direct summands: $\Psi_k(A) = A_k$ when $A \simeq \oplus_{k=1}^M A_k$. Then $\Psi_k(V_i) \in \mathcal{N}_k$, and by Lemma \ref{imageunderhomos} $U_k := \sum_{i=1}^p K_i \otimes \Psi_k(V_i) \in M_n(\C)\otimes \mathcal{N}_k$ is unitary and yields a factorization for the channel $\Phi_k$ determined by $Q_k^*Q_k = (\mathrm{id}\otimes \mathrm{Tr}_{\mathcal{N}_k})\bigl(\sum_{i,j} E_{ij} \otimes \Psi_k(V_i^*)\Psi_k(V_j)\bigr)$.\\
Finally, since $\mathrm{Tr}_{\mathcal{N}}(V_i^*V_j ) = \sum_{k=1}^M q_k \mathrm{Tr}_{\mathcal{N}_k}(\Psi_k(V_i^*)\Psi_k(V_j))$ we have that 
\begin{align*} I_p &= \sum_{i,j=1}^p E_{ij}\otimes \mathrm{Tr}_{\mathcal{N}}(V_i^*V_j) \\
& = \sum_{k=1}^M \sum_{i,j=1}^p E_{ij}\otimes \mathrm{Tr}_{\mathcal{N}_k}(\Psi_k(V_i^*)\Psi_k(V_j)) \\
& = \sum_{k=1}^M q_kQ_k^*Q_k.\end{align*}

Since each $\Phi_k$ has Choi matrix $C_{\Phi_k} = K(Q_k^*Q_k)^TK^*$, using the notation from Theorem \ref{choixtmpts} we see that
$$C_{\Phi} = KK^* = K(\sum_{k=1}^M q_k(Q_k^*Q_k)^T)K =\sum_{k=1}^m q_k C_{\Phi_k}$$
and hence $\Phi=\sum_{k=1}^M q_k \Phi_k$.
\end{proof}

\begin{remark} Notice that the above proof can fairly easily be modified to deal with direct integrals, with the caveat that the map from $\mathcal{N}$ to each of its direct integrands is a $*$-homomorphism almost everywhere.
\end{remark}
Since every von Neumann algebra is a direct integral of factors \cite{von}, we have proven the following:
\begin{cor}
If $\Phi$ is factorizable by an algebra $\mathcal{N}$, either $\mathcal{N}$ is a factor, or $\Phi$ is a convex combination of channels $\Phi_k$ each of which is factorized by a factor $\mathcal{N}_k$.\\
In particular, this proves that if $\Phi$ is factorizable by a direct integral of factors of type $I$, it is matrix factorizable; since if $\Phi$ is in the convex hull for $\{\Phi_k\}_{k\in X}$ for $k$ some measure space $X$, by Carath{e}odory's theorem we can choose a finite subset of these points and then $\Phi$ is a convex combination of a finite number of matrix factors.  
\end{cor}

In the case of matrix factorizations, we see that either $\Phi$ is factorizable by a factor $M_k(\C)$ for some $k$, or $\Phi$ is a convex combination of channels $\Phi_k$ each of which can be factored by a factor of type $I_{i_k}$.\\
The distinction between factorization by factors and by direct sums of factors allows us to formulate a sufficient condition for testing when a channel is extreme in the set of factorizable channels. To do so, recall that a matrix factorization by a factor of type $I_k$ corresponds to $A \in \mathcal{D}_Z(k)$ with $\mathrm{rank}(L_Z(A)) \leq k$ and $\mathrm{Tr}(A) = 0$.
\begin{cor}\label{sffcnt4xtm} Let $\Phi$ be a quantum channel that is matrix factorizable. Suppose the only solutions to $A\in \mathcal{D}_Z(k)$ and $\mathrm{rank}(L_Z(A))\leq k$ satisfy $\mathrm{Tr}(A) = 0$. Then $\Phi$ is extreme in the set of matrix factorizable channels. 
\end{cor}

In the case of Schur product channels, this condition can be stated a little more cleanly, owing to the observation that in this case, for $A \in \mathcal{D}_Z(k)$, $\mathrm{rank}(L_Z(A)) \geq k$. To see that this is true, recall that $\{Z_i\}_{i=1}^d$ are a basis for $\{w_iw_i^*\}^{\perp n}_{i=1}$ where $\{w_i\}$ are Gram vectors for the correlation matrix $C$. We can, without loss of generality, apply a unitary so that $w_1 = e_1$, hence $Z_{i_{11}} = 0$ for all $i$, and so 
$$L_Z(A) = I_p\otimes I_k + \sum_{i=1}^d Z_i\otimes A_i$$ has $I_k$ as its $(1,1)$ block.\\
Hence, in the case of a Schur product channel, Corollary \ref{sffcnt4xtm} can be phrased as, $\Phi$ factorizable is extreme in the factorizable correlation matrices if the only $A \in \mathcal{D}_Z(k)$ satisfying $\mathrm{rank}(L_Z(A)) = k$ satisfy $\mathrm{Tr}(A) = 0$.

\begin{example} The following example is borrowed from Haagerup and Musat \cite{haagerup}, where they proved that the correlation matrix $$ C = 
\begin{pmatrix} 1 & \beta & \beta & \beta & \beta & \beta \\
\beta & 1 & \beta & -\beta & -\beta & -\beta \\
\beta & \beta & 1 & \beta & -\beta & -\beta \\
\beta & -\beta & \beta & 1 & \beta & -\beta \\
\beta & -\beta & -\beta & \beta & 1 & \beta \\
\beta & -\beta & -\beta & -\beta & \beta & 1 \end{pmatrix}, \beta = \frac{1}{\sqrt{5}}$$ is factorizable, but is not random unitary. We will now show that in fact, it cannot be written as a convex combination of factorizable channels in any non-trivial way.
A set of Gram vectors for this matrix is 
\begin{align*}w_1 &= (1,0,0)^T \\ 
w_2 &= \frac{1}{\sqrt{5}}(1,\sqrt{2},\sqrt{2})^T \\ 
w_3 &= \frac{1}{\sqrt{5}}(1,\sqrt{2}\omega,\sqrt{2}\omega^4)^T \\ 
w_4 &= \frac{1}{\sqrt{5}}(1,\sqrt{2}\omega^2,\sqrt{2}\omega^3)^T \\
w_5 &=\frac{1}{\sqrt{5}}(1,\sqrt{2}\omega^3,\sqrt{2}\omega^2)^T\\ 
w_6 &= \frac{1}{\sqrt{5}}(1,\sqrt{2}\omega^4,\sqrt{2}\omega)^T\end{align*} 
for $\omega$ a primitive fifth root of unity. Then a basis for $\mathrm{ker}(\Phi^{C\dagger})$ is 
$$Z_1 = \begin{pmatrix} 0 & 0 & 0 \\ 0 & \sqrt{2} & 0 \\ 0 & 0 & -\sqrt{2} \end{pmatrix}, Z_2 = \begin{pmatrix} 0 & 1 & -1 \\ 1 & 0 & 0 \\ -1 & 0 & 0 \end{pmatrix}, Z_3 = i\begin{pmatrix} 0 & 1 & 1 \\ -1 & 0 & 0 \\ -1 & 0 & 0 \end{pmatrix}$$ since $w_i^*Z_jw_i = 0$ for all $i,j$. 
Then 
$$L_Z(A) =  I_3\otimes I_k + \sum_{i=1}^3 Z_i\otimes A_i = \begin{pmatrix} I_k & A_2 + iA_3 & -A_2 + iA_3 \\ A_2-iA_3 & I_k +\sqrt{2}A_1 & 0 \\ -A_2-iA_3 & 0 & I_k -\sqrt{2}A_1\end{pmatrix}$$ which has rank $k$ if and only if the Schur complement 
$$\begin{pmatrix} I_k + \sqrt{2}A_1 & 0 \\ 0 & I_k -\sqrt{2}A_1 \end{pmatrix} - \begin{pmatrix} A^* \\ -A \end{pmatrix} I_k^{-1}\begin{pmatrix} A & -A^*\end{pmatrix} = \begin{pmatrix} I_k + \sqrt{2}A_1 - A^*A & A^{*2} \\ A^2 & I_k - \sqrt{2}A_1 - AA^*\end{pmatrix}$$ is $0$, where $A = A_2 +iA_3$. 
This gives the equations 
\begin{align*}I_k + \sqrt{2}A_1 &= A^*A\\ 
I_k - \sqrt{2}A_1 &= AA^*\\
A^2 = A^{*2}&=0.\end{align*} 
Since $A$ is nilpotent, $\mathrm{Tr}(A)=0$, and hence also 
$$\mathrm{Tr}(A_2) = \frac{1}{2}(\mathrm{Tr}(A)+\mathrm{Tr}(A^*))=0$$ 
and 
$$\mathrm{Tr}(A_3) = \frac{1}{2i}(\mathrm{Tr}(A)-\mathrm{Tr}(A^*)) = 0.$$ 
Finally, $A_1 = \frac{1}{2\sqrt{2}}(A^*A-AA^*)$, and so also $$\mathrm{Tr}(A_1)=0.$$\\
By Corollary \ref{sffcnt4xtm}, the channel $\Phi(X)=X\circ C$ is an extreme point in the set of factorizable Schur product channels.
\end{example}
 \vspace{0.1in}

{\noindent}{\it Acknowledgements.} J.L. acknowledges support from the  the National Research Foundation of South Africa, NRF CPRR grant number 90551.

\bibliographystyle{plain}
\bibliography{schur}

\end{document}